\documentclass[preprint,12pt]{elsarticle}

\usepackage{amssymb}
\usepackage{latexsym,amssymb}
\usepackage{enumerate, verbatim}
\usepackage{graphics}
\usepackage[latin1]{inputenc}
\usepackage[T1]{fontenc}
\usepackage{amsmath, amsfonts, amsthm, amscd}

\newtheorem{theorem}{Theorem}[section]
\newtheorem{teo}[theorem]{Theorem}

\newtheorem{lemma}[theorem]{Lemma}
\newtheorem{coro}[theorem]{Corollary}

\newtheorem{remark}[theorem]{Remark}

\newtheorem{definition}[theorem]{Definition}

\newcommand{\CC} {\ensuremath{\mathcal{C}}}

\newcommand{\F} {\ensuremath{\mathbb{F}}}

\newcommand{\Aut} {\textnormal{\textrm{Aut}}}

\journal{Finite Fields and their Applications}

\begin{document}

\begin{frontmatter}

\title{The automorphism group of a self-dual $[72,36,16]$ code is not an elementary abelian group of order $8$}

\author{Martino Borello}

\address{Member
INdAM-GNSAGA (Italy), IEEE\\

 Dipartimento di Matematica e
Applicazioni\\ Universit\`{a} degli Studi di Milano Bicocca\\
20125 Milan, Italy\\ e-mail: m.borello1@campus.unimib.it}

\begin{abstract}
The existence of an extremal self-dual binary linear code $\CC$ of
length $72$ is a long-standing open problem. We continue the
investigation of its automorphism group: looking at the combination
of the subcodes fixed by different involutions and doing a computer
calculation with {\sc Magma}, we prove that $\Aut(\CC)$ is not
isomorphic to the elementary abelian group of order $8$. Combining
this with the known results in the literature one obtains that
$\Aut(\CC)$ has order at most $5$.

\end{abstract}


\end{frontmatter}

\section{Introduction}

A binary linear code of length $n$ is a subspace of $\F_2^n$, where
$\F_2$ is the field with $2$ elements. A binary linear code
$\mathcal{C}$ is called \emph{self-dual} if
$\mathcal{C}=\mathcal{C}^\perp$ with respect to the Euclidean inner
product. It follows immediately that the dimension of such a code
has to be the half of the length. The \emph{minimum distance} of
$\mathcal{C}$ is defined as
$\textnormal{d}(\mathcal{C}):=\min_{c\in\mathcal{C}\setminus\{0\}}\{\#\{i
\ | \ c_i=1\}\}$. In \cite{MSmindis} an upper bound for the minimum
distance of self-dual binary linear codes is given. Codes achieving
this bound are called \textit{extremal}. The most interesting codes,
for various reasons, are those whose length is a multiple of $24$:
in this case $\textnormal{d}(\mathcal{C})=4m+4$, where $24m$ is the
length of the code, and they give rise to beautiful combinatorial
structures \cite{AMdes}. There are unique extremal self-dual codes
of length $24$ (the extended binary Golay code $\mathcal{G}_{24}$)
and $48$ (the extended quadratic residue code $\mathcal{QR}_{48}$).
For nearly forty years many people have tried unsuccessfully to find
an extremal self-dual code of length $72$ \cite{S}. The usual
approach to this problem is to study the possible automorphism
groups (see next section for the detailed definition of it). Most of
the subgroups of $S_{72}$ are now excluded: the last result is
contained in \cite{BDN}, in which the authors finished to exclude
all the non-abelian groups with order greater than $5$.

In this paper we prove that the elementary abelian group of order
$8$ cannot occur as automorphism group of such a code, obtaining the
following.

\begin{theorem}
The automorphism group of a self-dual $[72,36,16]$ code is either
cyclic of order $1,2,3,4,5$ or elementary abelian of order $4$.
\end{theorem}

The techniques which we use are similar to those of \cite{Baut6}. We
know \cite{Neven}, up to equivalence, the possible subcodes fixed by
all the non-trivial involutions. So we combine them pairwise,
checking the minimum distance to be $16$, and we classify their sum,
up to equivalence. We get only a few extremal codes and all of them
satisfy certain intersection properties that, with easy dimension
arguments, make it impossible to sum a third fixed subcode without
loosing the extremality.

All results are obtained using extensive computations in {\sc Magma}
\cite{Magma}.

\section{Basic definitions and notations}

Throughout the paper we will use the following notations for groups:
\begin{itemize}
  \item $C_m$ is the cyclic group of order $m$;
  \item $S_m$ is the symmetric group of degree $m$;
  \item if $A$ and $B$ are two groups, $A\times B$ indicates their direct
  product;
  \item if $A$ and $B$ are two groups, $A\wr B$ indicates their wreath
  product.
\end{itemize}
Given a group $G$ and a subgroup $H$ of $G$ ($H\leq G$) we denote
$\textnormal{C}_G(H)$ the centralizer of $H$ in $G$. Let $\kappa\in
G$. Then
$\textnormal{C}_G(\kappa):=\textnormal{C}_G(\langle\kappa\rangle)$,
where $\langle\kappa\rangle$ is the (cyclic) group generated by
$\kappa$.

Let us consider the ambient space $\F_2^n$. We will indicate with
calligraphic capital letters the subspaces of $\F_2^n$, in order to
distinguish them from groups. We have a natural (right) action of
$S_n$ on $\F_2^n$ defined as follows: let $v=(v_1,\ldots,v_n)\in
\F_2^n$ and $\sigma\in S_n$; then
$$v^\sigma:=(v_{1^{\sigma^{-1}}},\ldots,v_{n^{\sigma^{-1}}}).$$
We have an action induced naturally on the subspaces of $\F_2^n$:
$$\mathcal{C}^\sigma:=\{c^\sigma \ | \ c\in\mathcal{C}\},$$
where $\CC\leq \F_2^n$ and $\sigma\in S_n$.

Let $\CC\leq \F_2^n$. Then the \emph{automorphism group} of the code
$\CC$ is the subgroup of $S_n$ defined as
$$\Aut(\CC):=\{\sigma\in S_n \ | \ \CC^\sigma=\CC\}.$$
Given a code $\CC$ and an automorphism $\sigma\in\Aut(\CC)$ we
define
$$\CC(\sigma):=\{c\in\CC \ | \ c^\sigma=c\}.$$
This is a subcode of $\CC$ and we call it the \emph{subcode fixed
by} $\sigma$.

\section{Preliminary observations}

Let ${\mathcal C}$ be a self-dual $[72,36,16]$ code such that
$$\Aut({\mathcal C})\cong C_2\times C_2\times C_2=\langle
\alpha,\beta,\gamma \rangle.$$
By \cite{Bord2} all non-trivial
elements of $\Aut(\CC)$ are fixed point free (that is of degree $n$)
and we may relabel the coordinates so that
$$\begin{array}{l}
\alpha = (1, 2)(3, 4)(5, 6)(7, 8) \ldots (71,72)   \\
\beta = (1, 3)(2, 4)(5, 7)(6, 8) \ldots (70,72)   \\
\gamma = (1, 5)(2, 6)(3, 7)(4, 8) \ldots (68,72).
\end{array}
$$

\begin{definition}
Let $\mathcal{V}:=\F_2^n$. Then
$$\begin{array}{l} \pi_\alpha:  \mathcal{V}(\alpha) \to \F_2^{36}  \\
(v_1,v_1,v_2,v_2,v_3,v_3,v_4,v_4, \ldots,v_{36},v_{36} ) \mapsto
 (v_1,v_2,\ldots , v_{36} ) \end{array}  $$
denote the bijection between the subspace of fixed by $\alpha$ and
$\F_2^{36}$,
$$\begin{array}{l} \pi_\beta:  \mathcal{V}(\beta) \to \F_2^{36}  \\
(v_1,v_2,v_1,v_2,v_3,v_4,v_3,v_4 \ldots,v_{35},v_{36} ) \mapsto
 (v_1,v_2,\ldots , v_{36} ) \end{array}  $$
denote the bijection between the subspace fixed by $\beta$ and
$\F_2^{36}$ and
$$\begin{array}{l} \pi_\gamma:  \mathcal{V}(\gamma) \to \F_2^{36}  \\
(v_1,v_2,v_3,v_4,v_1,v_2,v_3,v_4 \ldots,v_{35},v_{36} ) \mapsto
 (v_1,v_2,\ldots , v_{36} ) \end{array}  $$
denote the bijection between the subspace fixed by $\gamma$ and
$\F_2^{36}$.
\end{definition}

\begin{remark}
The centralizer $\textnormal{C}_{S_{72}}(\alpha) \cong C_2 \wr
S_{36}$ of $\alpha$ acts on the set of fixed points of $\alpha$.
Using the isomorphism $\pi_\alpha$ we hence obtain a group
epimorphism which we denote by $\eta_\alpha$
$$\eta_\alpha : \textnormal{C}_{S_{72}}(\alpha) \to S_{36}$$
with kernel $C_2^{36}$. Similarly we obtain the epimorphisms
$$\eta_\beta : \textnormal{C}_{S_{72}}(\beta) \to S_{36}$$
and
$$\eta_\gamma : \textnormal{C}_{S_{72}}(\gamma) \to S_{36}.$$
\end{remark}

By \cite{Neven} we have that all the projections of the fixed codes
$\pi_\alpha(\CC(\alpha))$,$\pi_\beta(\CC(\beta))$ and
$\pi_\gamma(\CC(\gamma))$ are self-dual $[36,18,8]$ codes. Such
codes have been classified in \cite{Gaborit}, up to equivalence
(under the action of the full symmetric group $S_{36}$) there are
$41$ such codes. Notice that
$$\langle \eta_\alpha(\beta),\eta_\alpha(\gamma)\rangle=\langle
\eta_\beta(\alpha),\eta_\beta(\gamma)\rangle = \langle
\pi_\gamma(\alpha),\eta_\gamma(\beta)\rangle=\langle \chi,\mu
\rangle\leq S_{36},$$ with
$$\chi=(1,2)(3,4)\ldots(35,36)$$
and
$$\mu=(1,3)(2,4)\ldots(34,36),$$
are contained in $\Aut(\pi_\alpha(\CC(\alpha)))$,
$\Aut(\pi_\beta(\CC(\beta)))$ and $\Aut(\pi_\gamma(\CC(\gamma)))$
respectively. Only $14$ of the $41$ codes, say
$\mathbb{Y}:=\{\mathcal{Y}_1,\ldots,\mathcal{Y}_{14}\}$, have an
automorphism group which contains at least one subgroup conjugate to
$\langle \chi,\mu \rangle$.

By direct calculation on these $14$ codes we get the following
conditions on the intersection of the codes.

\begin{lemma}\label{cases}
Let
$$(\chi',\mu',\zeta')\in\{(\alpha,\beta,\gamma),(\alpha,\gamma,\beta),(\beta,\alpha,\gamma),(\beta,\gamma,\alpha),(\gamma,\alpha,\beta),(\gamma,\beta,\alpha)\}.$$
Then we have only the following possibilities:
$$\begin{tabular}{|c|c|c|} \hline
$\dim(\CC(\chi')\cap\CC(\mu')\cap\CC(\zeta'))$ &
$\dim(\CC(\chi')\cap\CC(\mu'))$ &
$\dim(\CC(\chi')\cap\CC(\zeta'))$ \\
  \hline
  $5$ & $9$ & $9$ \\
  $5$ & $9$ & $10$ \\
  $6$ & $9$ & $9$ \\
  $6$ & $9$ & $10$ \\
  $6$ & $9$ & $11$ \\
  $6$ & $10$ & $10$ \\
  $6$ & $10$ & $11$ \\
  \hline
\end{tabular}$$
\end{lemma}

Let $G:= \textnormal{C}_{S_{72}}(\Aut(\CC))$. Then $G$ acts on the
set of extremal self-dual codes with automorphism group $\langle
\alpha,\beta,\gamma\rangle$ and we aim to find a system of orbit
representatives for this action. Here we have some differences with
the non-abelian cases, since the full group $\langle
\alpha,\beta,\gamma\rangle$ is a subgroup of the automorphism group
of all the fixed subcodes $\CC(\alpha),\CC(\beta)$ and
$\CC(\gamma)$. The main property that we use is the following, which
is straightforward to prove:
\begin{equation}\label{inter}
\pi_\alpha(\CC(\alpha))(\chi)=\pi_\beta(\CC(\beta))(\chi)
\end{equation}
and similar relations for the other fixed subcodes. This allows us
to combine properly $\CC(\alpha)$ and $\CC(\beta)$ classifying their
sum.

\section{Description of the calculations}

Let $$\mathbb{D}:=\{\mathcal{D}=\mathcal{D}^\perp\leq \F_2^{36} \ |
\ \textnormal{d}(\mathcal{D})=8, \ \langle \chi,\mu\rangle \leq
\Aut(\mathcal{D}) \}.$$
The group
$$G_{36}:=\textnormal{C}_{S_{36}}(\langle
\chi,\mu\rangle)=\eta_\alpha(G)=\eta_\beta(G)=\eta_\gamma(G)$$ acts,
naturally, on this set.

\begin{lemma}\label{rep}
A set of representatives of the $G_{36}$-orbits on $\mathbb{D}$ can
be computed by performing the following computations on each
$\mathcal{Y}\in\mathbb{Y}$:
\begin{itemize}
\item Let $\chi_1,\ldots,\chi_{s}$ represent the
conjugacy classes of fixed point free elements of order $2$ in
$\Aut(\mathcal{Y})$.
\item Compute elements $\tau_1,\ldots,\tau_{s} \in S_{36}$ such that
$\tau_k^{-1}\chi_k\tau_k =\chi$ and put $\mathcal{Y}_{k} :=
\mathcal{Y}^{\tau_k}$ so that $\chi \in \Aut (\mathcal{Y}_{k})$.
\item For every $\mathcal{Y}_k$, consider the set of fixed point free elements $\tilde{\mu}$ of order $2$ in
$\textnormal{C}_{\Aut(\mathcal{Y}_k)}(\chi)$ such that $\langle
\chi,\tilde{\mu}\rangle$ is conjugate to $\langle \chi,\mu\rangle$
in $S_{36}$. Let $\mu_1,\ldots,\mu_{t_k}$ represent the
$\textnormal{C}_{\Aut(\mathcal{Y}_k)}(\chi)$-conjugacy classes in
this set.
\item Compute elements $\sigma_1,\ldots,\sigma_{t_k} \in \textnormal{C}_{S_{36}}(\chi)$ such that
$\sigma_l^{-1}\mu_l\sigma_l =\mu$ and put $\mathcal{Y}_{k,l} :=
\mathcal{Y}_k^{\sigma_l}$ so that $\langle \chi,\mu\rangle \leq \Aut
(\mathcal{Y}_{k,l})$.
\end{itemize}
Then $\mathbb{D}':=\{ \mathcal{Y}_{k,l} \mid
\mathcal{Y}\in\mathbb{Y}, 1\leq k \leq s, 1\leq l \leq t_k \} $
represents the $G_{36}$-orbits on ${\mathbb D}$.
\end{lemma}

\begin{proof}
Clearly these codes lie in ${\mathbb D}$. \\
Since $G_{36}\leq S_{36}$, if we consider different elements in
$\mathbb{Y}$, say $\mathcal{Y}$ and $\mathcal{Y}'$, then
$\mathcal{Y}'_{k',l'}$ is not in the same orbit of
$\mathcal{Y}_{k,l}$ for any
$k',l',k,l$.\\
Now assume that there is some $\lambda \in G_{36}$ such that
$$\mathcal{Y}^{\tau _{k'}\sigma_{l'}} = \mathcal{Y}_{k',l'} ^{\lambda } = \mathcal{Y}_{k,l}
= \mathcal{Y}^{\tau _{k} \sigma_l  }.$$ Then
$$\epsilon :=  \tau _{k'} \sigma_{l'} \lambda \sigma_{l} ^{-1} \tau _{k} ^{-1}
\in \Aut (\mathcal{Y}) $$ satisfies $\epsilon \chi_k \epsilon ^{-1}
= \chi_{k'} $, so $\chi_k$ and $\chi_{k'}$ are conjugate in $\Aut
(\mathcal{Y})$, which implies $k=k'$ (and so $\tau_k=\tau_{k'}$).
Now,
$$\mathcal{Y}^{\tau _{k} \sigma_{l'} \lambda } =\mathcal{Y}_{k}^{\sigma_{l'}
\lambda}= \mathcal{Y}_k^{\sigma_{l}  } = \mathcal{Y}^{\tau_{k}
\sigma_l}.$$ Then
$$\epsilon' :=  \sigma_{l'} \lambda \sigma_{l}^{-1}
\in \Aut (\mathcal{Y}_k)$$ commutes with $\chi$. Furthermore
$\epsilon' \sigma_l {\epsilon'}^{-1} = \sigma_{l'} $ and hence
$l=l'$.

Now let $\mathcal{Z} \in {\mathbb D}$ and choose some $\xi \in
S_{36}$ such that $\mathcal{Z}^{\xi } = \mathcal{Y}\in\mathbb{Y} $.
Then $\xi^{-1}\chi\xi$ is conjugate to some of the chosen
representatives $\chi_k \in \Aut(\mathcal{Y})$ ($i=1,\ldots ,s$) and
we may multiply $\xi $ by some automorphism of $\mathcal{Y}$ so that
$$\xi^{-1}\chi\xi = \chi_k = \tau_k \chi  \tau_k^{-1}.$$ So $\xi
\tau_k \in \textnormal{C}_{S_{36}} (\chi)$ and $\mathcal{Z} ^{\xi
\tau _k } = \mathcal{Y} ^{\tau _k}=\mathcal{Y}_k $.\\ It is
straightforward to prove that the element $(\xi \tau _k)^{-1}\mu(\xi
\tau _k) \in \Aut (\mathcal{Y}_k)$ is a fixed point free element of
order $2$ in $\textnormal{C}_{\Aut(\mathcal{Y}_k)}(\chi)$ such that
$\langle \chi,(\xi \tau _k)^{-1}\mu(\xi \tau _k) \rangle$ is
conjugate to $\langle \chi,\mu\rangle$ in $S_{36}$. So there is some
automorphism $\omega \in \textnormal{C}_{\Aut(\mathcal{Y}_k)}(\chi)$
and some $l \in \{1,\ldots , t_k \}$ such that $(\xi \tau
_k\omega)^{-1}\mu(\xi \tau _k\omega) = \mu_l$. Then
$$\mathcal{Y}^{\xi \tau_k \omega \sigma_l}=\mathcal{Y}_{k,l} $$ where
$\xi \tau_k \omega \sigma_l \in G_{36}$.
\end{proof}

There are $242$ such representatives. For our purposes we need to
modify this set a little: consider the set $\{\mathcal{Y}(\chi)\ | \
\mathcal{Y}\in \mathbb{D}\}$ and take a set of representatives for
the action of $G_{36}$ on this set, say
$\mathbb{E}:=\{\mathcal{E}_1,\ldots,\mathcal{E}_m\}$. By
calculations $m=40$. For every $1\leq i\leq m$ define the set
$$\tilde{\mathbb{D}}_i:=\{\mathcal{Y}^\epsilon \ | \ \mathcal{Y} \in
\mathbb{D}' \ \text{such that there exists} \ \epsilon \in G_{36} \
\text{so that} \ \mathcal{Y}(\chi)^\epsilon =\mathcal{E}_i\}.$$
Clearly $\bigcup_{i=1}^m \tilde{\mathbb{D}}_i$ is still a set of
representatives of the  $G_{36}$-orbits on ${\mathbb D}$, but now
$\mathcal{Y}_j(\chi)$ and $\mathcal{Y}_k(\chi)$ are equal if
$\mathcal{Y}_j$ and $\mathcal{Y}_k$ belong to the same
$\tilde{\mathbb{D}}_i$ and they are not equivalent via the action of
$G_{36}$ if $\mathcal{Y}_j$ and $\mathcal{Y}_k$ do not belong to the
same $\tilde{\mathbb{D}}_i$.

Let
$$\mathbb{D}_{{(\alpha,\beta)}_i}=\{\pi_\alpha^{-1}(\mathcal{Y}_\alpha)+(\pi_\beta^{-1}(\mathcal{Y}_\beta))^\omega \leq \F_2^{72}
\ | \ \mathcal{Y}_\alpha,\mathcal{Y}_\beta\in \tilde{\mathbb{D}}_i,
\omega\in \textnormal{C}_{\Aut(\mathcal{Y}_\beta(\chi))}(\langle
\chi,\mu\rangle)\}.$$

\begin{remark}
Considering $(\pi_\beta^{-1}(\mathcal{Y}_\beta))^\omega$ with
$\omega$ varying in
$\textnormal{C}_{\Aut(\mathcal{Y}_\beta(\chi))}(\langle
\chi,\mu\rangle)$ is exactly the same as considering
$(\pi_\beta^{-1}(\mathcal{Y}_\beta))^\tau$ with $\tau$ varying in a
right transversal of $$\Aut(\mathcal{Y}_\beta(\chi)) \cap
\textnormal{C}_{\Aut(\mathcal{Y}_\beta(\chi))}(\langle
\chi,\mu\rangle)$$ in
$$\textnormal{C}_{\Aut(\mathcal{Y}_\beta(\chi))}(\langle
\chi,\mu\rangle).$$ Obviously this makes the calculations faster.
\end{remark}

\begin{lemma}
The code $\CC(\alpha)+\CC(\beta):=\{v+w \ | \ v\in \CC(\alpha) \
\text{and} \ w \in \CC(\beta) \}$ is equivalent, via the action of
$G$, to an element of $\bigcup_{i=1}^m \mathbb{D}_{{(a,b)}_i}$.
\end{lemma}

\begin{proof}
By Lemma \ref{rep} and by construction of $\bigcup_{i=1}^m
\tilde{\mathbb{D}}_i$, there exist $i\in\{1,\ldots,m\}$,
$\mathcal{Y}_a\in \tilde{\mathbb{D}}_i$ and $\bar{\rho}\in G_{36}$
such that $\pi_\alpha(\CC(\alpha))^{\bar{\rho}}=\mathcal{Y}_\alpha$.
Choose $\rho\in\eta_\alpha^{-1}(\overline{\rho})$. Then it is easy
to observe that
\begin{itemize}
\item $\pi_\beta(\CC^\rho(\beta))$ is a self-dual $[36,18,8]$ code;
\item $\langle \chi,\mu\rangle \leq \Aut (\pi_\beta(\CC^\rho(\beta)))$ (since
$\rho \in G$);
\item $(\pi_\beta(\CC^\rho(\beta)))(\chi)=(\pi_\alpha(\CC^\rho(\alpha)))(\chi)=\mathcal{E}_i$ (as in \eqref{inter}).
\end{itemize}
Now, $\{(\mathcal{Y}_\beta)^\tau \ | \ \mathcal{Y}_\beta \in
\tilde{\mathbb{D}}_i, \ \tau\in
\textnormal{C}_{\Aut(\mathcal{Y}_\beta(\chi))}(\langle \chi,\mu
\rangle)\}$ is the set of all possible such codes, so
$(\pi_\beta(\CC^\rho(\beta)))(\chi)$ is one of these codes.
\end{proof}

\begin{remark}
There are, up to equivalence in the full symmetric group $S_{72}$,
only $22$ codes in $\bigcup_{i=1}^m \mathbb{D}_{{(\alpha,\beta)}_i}$
such that the minimum distance is at least $16$, say
$\mathcal{D}_1,\ldots,\mathcal{D}_{22}$. They are all $[72,26,16]$
codes. In particular
$$\dim(\mathcal{D}_i(\alpha)\cap \mathcal{D}_i(\beta))=10.$$
\end{remark}

\begin{coro}
The code $\CC(\alpha)+\CC(\beta)$ is equivalent, via the action of
the full symmetric group $S_{72}$, to a code $\mathcal{D}_i$, with
$i\in\{1,\ldots,22\}$.
\end{coro}

We can repeat in a completely analogous way all the procedure for
the pairs $(\alpha,\gamma)$ and $(\beta,\gamma)$, interchanging the
roles of the elements $\alpha,\beta$ and $\gamma$. Then we get the
following.

\begin{coro}
The codes $\CC(\alpha)+\CC(\gamma):=\{v+w \ | \ v\in \CC(\alpha) \
\text{and} \ w \in \CC(\gamma) \}$ and
$\CC(\beta)+\CC(\gamma):=\{v+w \ | \ v\in \CC(\beta) \ \text{and} \
w \in \CC(\gamma) \}$ are equivalent, via the action of the full
symmetric group $S_{72}$, to some codes $\mathcal{D}_j$ and
$\mathcal{D}_k$, with $j,k\in\{1,\ldots,22\}$.
\end{coro}

This implies that
\begin{equation}\label{eq1}\dim(\CC(\alpha)\cap
\CC(\gamma))=10 \quad \text{and} \quad \dim(\CC(\beta)\cap
\CC(\gamma))=10.
\end{equation}
Furthermore, by {\sc Magma} calculations we get that
\begin{equation}\label{eq2} \dim(\CC(\alpha) \cap \CC(\beta)\cap \CC(\gamma))=5.
\end{equation}
Both statements can be verified by taking all the elements
$\alpha',\beta',\gamma'$ of order $2$ and degree $72$ in
$\Aut(\mathcal{D}_i)$ such that $\langle
\alpha',\beta',\gamma'\rangle$ is conjugate to $\langle
\alpha,\beta,\gamma\rangle$ in $S_{72}$.

To get a contradiction it is now enough to observe that \eqref{eq1}
and \eqref{eq2} are not compatible with the table in Lemma
\ref{cases}. So we conclude the following.

\begin{teo}
The automorphism group of a self-dual $[72,36,16]$ code does not
contain a subgroup isomorphic to $C_2\times C_2 \times C_2$.
\end{teo}

\section*{Acknowledgment}

The author expresses his gratitude to F. Dalla Volta and G. Nebe for
the fruitful discussions and suggestions. \emph{Laboratorio di
Matematica Industriale e Crittografia} of Trento deserves thanks for
the help in the computational part.

\end{document}